\documentclass[a4paper,reqno]{amsart}
\usepackage{geometry}
\usepackage{fullpage}
\usepackage[latin1]{inputenc}
\usepackage[T1]{fontenc}
\usepackage{amsfonts}
\usepackage{amssymb}
\usepackage{amsmath}
\usepackage{mathrsfs}
\usepackage{amsthm}
\usepackage{graphicx}
\usepackage[dvipsnames]{xcolor}
\usepackage{afterpage}
\usepackage[colorlinks=true, linkcolor=magenta, citecolor=cyan, urlcolor=green]{hyperref} %%questo fa i link colorati
\usepackage{bbm}
\usepackage{latexsym}

\newtheorem{proposition}{Proposition}
\newtheorem{lemma}{Lemma}

\newtheorem{remark}{Remark}

\numberwithin{equation}{section}

\definecolor{light}{gray}{.9}

\def\be{\begin{equation}}
\def\ee{\end{equation}}
\def\bea{\begin{eqnarray}}
\def\eea{\end{eqnarray}}
\def\ni{\noindent}
\def\nn{\nonumber}
\def\s{\sigma}

\def\a{\alpha}

\def\d{\delta}

\def\b{\beta}

\def\l{\lambda}

\def\t{\tau}

\def\R{\mathbb{R}}
\def\E{\mathbb{E}}
\def\SSS{\mathscr{S}}

\def\diag{\operatorname{diag}}
\def\Ker{\operatorname{Ker}}

\newcommand{\ie}{\textit{i.e. }}

\newcommand{\meanv}[1]{\left\langle#1\right\rangle}

\newcommand{\nocontentsline}[3]{}
\newcommand{\tocless}[2]{\bgroup\let\addcontentsline=\nocontentsline#1{#2}\egroup}

\DeclareMathSymbol{\leqslant}{\mathalpha}{AMSa}{"36} % nicer `smaller or equal'
\DeclareMathSymbol{\geqslant}{\mathalpha}{AMSa}{"3E} % nicer `larger or equal'
\DeclareMathSymbol{\eset}{\mathalpha}{AMSb}{"3F}     % nicer `emptyset'
\renewcommand{\leq}{\;\leqslant\;}                   % redef. of < or =
\renewcommand{\geq}{\;\geqslant\;}                   % redef. of > or =

\newcommand{\OOO}[1]{O \left(#1\right)}

\def\Var{\operatorname{Var}}

\title{Non--Convex Multipartite Ferromagnets}
\date{\today}

\author{Giuseppe Genovese}
\address{Giuseppe Genovese: Institut f\"ur Mathematik, Universit\"at Z\"urich,
Winterthurerstrasse 190, CH-8057 Z\"urich, Switzerland.}
\email{giuseppe.genovese@math.uzh.ch}

\author{Daniele Tantari}
\address{Daniele Tantari: Centro Ennio de Giorgi, Scuola Normale Superiore, Piazza dei Cavalieri 3, I-56100 Pisa (Italy).}
\email{daniele.tantari@sns.it}

\begin{document}
\maketitle

\begin{abstract}

\ni We investigate a multipartite ferromagnetic model without self-interactions between spins of the same party, so that the Hamiltonian is not a definite quadratic form of the magnetisations. We find the free energy and study the phase transition for all zero external fields. Moreover in the bipartite case we analyse the fluctuations of the rescaled magnetisations, below and at the critical point, and we study the phase transitions with non-zero magnetic fields.

\vspace{5mm}

\end{abstract}

\section{Introduction}\label{intro}

The mean field ferromagnet, or Curie-Weiss model (CW henceforth),  is a paradigm in statistical physics, as a very simple model exhibiting the elemental phenomenon of spontaneous symmetry breaking. For this reason it has been intensively studied and used as a toy model to test many methods in equilibrium statistical mechanics. 

We are concerned here with a variation of the usual ferromagnets, namely the CW models defined on a multipartite graph. Such systems model quite naturally any situation in which there are (finitely) many different agents involved, positively correlated. It is noteworthy that the bipartite ferromagnet provides a good mean field approximation for Ne\'el ferrimagnetism, as shown in \cite{BS80}.

Therefore, albeit the main motivation of our investigation comes from the attention given recently to multipartite spin glasses \cite{bgg,ahp,pan}, multipartite ferromagnets are by themselves of a certain interest and, mostly in the bipartite case, they have been already studied in \cite{gallo,gb2,bgg,fedcon,fedele1,fedele2,pizzo}. 

In this paper we deal with Hamiltonians which, as quadratic forms in the magnetisations, are not of definite sign. This feature brings some new issues in the investigation of multipartite systems, no matter if disordered or ferromagnetic. Therefore, many of the ideas used in this paper can find a suitable analogue in the analysis of the (much more complicate) disordered models.

\bigskip

We consider $\nu\geq2$ sets of $N_a$, $a=1,\dots,\nu$ spin variables. We denote by $\sigma^a_i$, $i=1, ..., N_a$ the variables in the $a- $th subset. Throughout the paper we will always make the following assumptions: all the spins are i.i.d r.vs, with symmetric probability distribution fulfilling 
\be\label{cond:CW-subGauss}
\begin{cases}
\Var_a[\s]=\E_{\s^a}[\s^2]=1\,,\\
\E_a[\s^4]-3<0\,,\\
r_{a}(t):=\E_{\s^a}[e^{t\s^2}]<\infty,\qquad\forall t>0\,.
\end{cases}
\ee
In principle the distributions could be different: we allow different statistics for each party. 
We also define the cumulant generating function as
\be\label{eq:phi}
\phi_a(t):=\log\E_{\s^a} [e^{t\s^a}]\,.
\ee
It is immediate to check that $\phi_a(t)$ is an even, analytic and uniformly convex function and 
\be\label{eq:espansione-phi}
\phi_a(t)=\frac{t^2}{2}+\sum_{k\geq 2} \frac{P^a_{2k}}{(2k)!} t^{2k}\,,
\ee
where $P^a_{2k}$ is the $2k$-th cumulant for the distribution of $\s^a$. We notice $P^a_4<0$, $a=1,\dots,\nu$.

%that $\E[\s^{2k}]\leq(2k-1)!!$ (the moments of the normal distribution, see also \cite{ellis}). In particular this entails $P^a_4=\E[\s^4]-3\leq0$. For simplicity we assume always $P^a_4\neq0$, $a=1,\dots,\nu$.

Let $N:=\sum_{a=1}^{\nu}N_a$ to be the total number of spin. We set (with a little abuse of notation) the relative size of the subset to be $N_a/N=\alpha_a\in(0,1)$. More precisely, to avoid a trivial behaviour of the model, we perform the 
thermodynamic limit $N\to\infty$ with the prescription
$\lim_N
\frac{N_a}{N}=:\alpha_a$.

\ni For each party labeled by $a=1,\cdots,\nu$ we can define the partial magnetisation $m_a$ as
\be\label{eq:m_x}
m_a:=\frac{1}{N_a}\sum_{i=1}^{N_a}\s^a_i.
\ee

\ni We let the spins interact via the Hamiltonian $H_N(\boldsymbol{\sigma}^1,\cdots,\boldsymbol{\s}^\nu)$:
\be\label{eq:H-multi}
H_N(\boldsymbol{\sigma}^1,\cdots,\boldsymbol{\s}^\nu):=-N\sum_{(a,b)}^{\nu}\a_a\a_b m_am_b
- N\sum_{a=1}^{\nu}h_a\a_a m_a.
\ee
Spins in each subsystem interact only with
spins in the other one, but not among themselves. The Hamiltonian (\ref{eq:H-multi}) (with all zero external fields) can be written as a quadratic form in the magnetisation vector as
\be\label{eq:non-convex}
H_N=-N(m,{\bf J}m)\,,
\ee
where the interaction matrix is
\be
J_{ab}:=
\begin{cases}
0 & a=b\,;\\
\a_a\a_b& a\neq b\,.\\
\end{cases}
\ee
This is not a definite quadratic form, that reflects the non-convex nature of the model.

Partition
function, pressure and free energy per site of the model are
defined as usual by
\begin{eqnarray}
Z_N(\beta, h_1,\cdots, h_\nu)&:=&\mathbb{E}_{\sigma^1}\cdots\mathbb{E}_{\sigma^\nu}e^{-\beta H_N(\boldsymbol{\sigma}^1,\cdots,\boldsymbol{\s}^\nu)}\,,\\
A_N(\beta, h_1,\cdots, h_\nu)&:=&\frac{1}{N}\log Z_N(\beta, h_1,\cdots, h_\nu)\,,\\
f_N(\beta, h_1,\cdots, h_\nu)&:=&-\frac{1}{\b}A_N(\beta, h_1,\cdots, h_\nu)\,.
\end{eqnarray}
Moreover we put
\be
A(\b,\a,h):=\lim_N A_N(\beta, h_1,\cdots, h_\nu)\,.
\ee
The Gibbs measure associated to this system is
$
Z_N^{-1}e^{-\b H_N}\,.
$
For any observable $O$ we denote with $\meanv{O}_N$ its mean value and with $\meanv{O}$ the same in the thermodynamic limit $N\to\infty$. A special role will be played by the partial magnetisations, which can be arranged in a vectorial form (in $\R^\nu$): we put $m:=(m_1,\dots,m_\nu)$ and analogously $\meanv{m}:=(\meanv{m_1},\dots,\meanv{m_\nu})$.

\bigskip

Our main achievements are listed as follows:

\bigskip

\begin{itemize}

\item[{\em i)}] We find the pressure of the model in terms of a variational principle (Proposition \ref{Prop:pressure}). This is a minimum principle and so it is reversed with respect of the usual ferromagnets. It can be formulated w.r.t. the real magnetisations (see (\ref{eq:A})) or w.r.t.  tilted order parameters(see (\ref{eq:A-m'})) and in the bipartite case it is equivalent to the $\min\max$ of \cite{gb2,bgg}.\\

\item[{\em ii)}] We study the spontaneous symmetry breaking for zero external fields, finding a multidimensional generalisation of the CW phase transition (see Propositions \ref{prop:b<bc} and \ref{prop:b>bc}).\\

\item [{\em iii)}] We find the fluctuations of the rescaled magnetisations in the bipartite model both in the paramagnetic phase (Proposition \ref{prop:Flutt-Bip}) and at the critical point (Proposition \ref{Prop:bip-fluct-crit}). \\

\item [{\em iv)}] We study with the aid of numerical simulations the phase diagram of the bipartite model as $h_1h_2<0$. We see the appearance of non trivial stable, unstable and metastable states as we tune the magnitude of the external fields (see Section \ref{SezioneBip}). 

\end{itemize}

\bigskip

The core of this paper is Section \ref{core}. There we compute the free energy of the model using the mapping into a Hamilton--Jacobi equation, combined with an appropriate comparison argument. The Hamilton--Jacobi method is quite powerful for this kind of models and it has a singular history, as it was discovered independently by several authors through the last thirty years: Brankov and Zagrebnov \cite{russi}, Newman \cite{new}, Choquard and Wagner \cite{losanna} and Genovese and Barra \cite{gb}. Moreover the work of Guerra \cite{sum}, which however deals with the replica symmetric solution of the Sherrington Kirkpatrick model rather than ferromagnetic models, definitely deserves to be mentioned (in the same direction, see also \cite{gino}). 

In our context however this technique cannot be applied tout-court. We need a suitable interpolation between the $\nu$-partite system and $\nu$ CW models at different temperatures. In this way we can make the interaction to be positive definite and thereby map the problem into a viscous Hamilton-Jacobi equation with convex (quadratic) Hamiltonian. The idea is already in \cite{bgg}.  The study of the phase transition is entirely analytical at zero external fields, as we face it as a bifurcation problem.  However this is not the sole situation in which a phase transition can occur. We discuss further this delicate point in Section \ref{Sect-Discussion}.

In Section \ref{SezioneBip} we focus on the properties of the bipartite model. This section relies very much on Appendix \ref{sec:genferr}. The key idea is to use the equivalence between the bipartite models and some generalised ferromagnetic models. This comes from the correspondence between sub-Gaussian probability measures and smooth convex functions, which stems from Bernstein's theory of completely monotonic functions (see \cite{fe}). 

In this way we find again the results of Section \ref{core} for $\nu=2$. Furthermore we establish ferromagnetic fluctuations of the rescaled magnetisations (as found for the CW model in \cite{SG}\cite{ellis}) for the bipartite model in the paramagnetic phase and at the critical point, by the analysis of those of the generalised ferromagnets. These results extend the achievements of \cite{fedcon}\cite{fedele2} to the non convex case. 

In the second part of Section \ref{SezioneBip} we analyse the bipartite model with both parties made of Bernoulli $\pm1$ and non zero external fields such that $h_1h_2<0$. By the self-consistent equations for the magnetisations, we can readily deduce two critical lines (Fig. \ref{fig1}), and we study the behaviour of the system when varying the fields along those lines. We find four different regimes (and three critical values for the field magnitudes), corresponding to different energy landscapes (Fig. \ref{fig2}, \ref{fig4} and \ref{fig5}). We investigate numerically the stability of the equilibrium states and the nature of the associated phase transitions.

Since the generalisation of the CW model to other convex interactions is not so immediate (albeit we find anything striking from the physical viewpoint), we give an account of it in Appendix \ref{sec:genferr}.

Finally some questions left open by our analysis are discussed in Section \ref{Sect-Discussion}.

\section{The Free Energy and Critical Behaviour}\label{core}

We have already emphasised that the Hamiltonian is not a definite quadratic form, that is the main difficulty of these models. Nonetheless we can make it so, by adding a suitably strong counterterm. Let $c>0$ and
\be\label{eq:Jc}
{\bf J}^c:= c\diag(\a_1^2,\dots,\a_\nu^2)-{\bf J}\,.
\ee
We have
\begin{lemma}\label{lemma:radice}
The quadratic form (\ref{eq:Jc}) is positive definite iff $c\geq\nu-1$. We have ${\bf J}^c={\bf P}^T{\bf P}$, where the rows of $\bf {P}$ are the linearly independent vectors 
\begin{eqnarray}\label{base}
v^1&=&\frac{\sqrt{c+1-\nu}}{\sqrt{\nu}}(\a_1,\dots,\a_\nu)\nonumber\\
v^2 &=&\frac{\sqrt{c+1}}{\sqrt{2}}(\a_1,-\a_2,0,\dots,0)\nonumber \\
\vdots&=&\qquad\vdots\nonumber \\
v^{a} &=& \frac{\sqrt{c+1}}{\sqrt{a(a-1)}}(\a_1,\dots,\a_{a-1},\a_a-a\a_a,0,\cdots,0)\nonumber\\
\vdots&=&\qquad\vdots\nonumber \\
v^{\nu} &=& \frac{\sqrt{c+1}}{\sqrt{\nu(\nu-1)}}(\a_1,\dots,\a_{\nu-1},\a_\nu-\nu\a_\nu)\,.
\end{eqnarray} 
\end{lemma}
\begin{proof}
Let ${\bf A}:=\diag(\a_1,\dots,\a_\nu)$. We see immediately that ${\bf J}^c={\bf A}{\bf T}^c{\bf A}$ with the matrix ${\bf T}^c$ defined by
\be
T^c_{ab}:=
\begin{cases}
c & a=b\,;\\
-1& a\neq b\,.\\
\end{cases}
\ee
Therefore we must look at ${\bf T}^c$. The secular equation reads
$$
(c+1-\lambda)w - \sum_a w_a=0\,.
$$
If $\sum_a w_a\neq 0$ then $w^1=\nu^{-1/2}(1,\dots,1)$ is constant and $\lambda= c+1-\nu$. We need to choose $c\geq \nu -1$ if we want $T$ to be positive defined. The other $\nu -1$ eigenvectors lie in the subspace $\sum_a w_a= 0$ and they all correspond to the eigenvalue $\lambda = c+1$. We see by a straightforward computation that the vectors
\begin{eqnarray}\label{base}
w^2 &=&\frac{\sqrt{c+1}}{\sqrt{2}}(1,-1,0,\dots,0)\nonumber \\
\vdots&\vdots&\qquad\vdots\nonumber \\
w^{a} &=& \frac{\sqrt{c+1}}{\sqrt{a(a-1)}}(1,\dots,1,1-a,0,\cdots,0)\nonumber\\
\vdots&\vdots&\qquad\vdots\nonumber \\
w^{\nu} &=& \frac{\sqrt{c+1}}{\sqrt{\nu(\nu-1)}}(1,\dots,1,1-\nu)\,
\end{eqnarray} 
form a basis of the subspace orthogonal to $w^1$. Let ${\bf P}'$ the $\nu\times\nu$ matrix with $a$-th rows $w_a$. Then ${\bf T}^c={\bf P}'^T{\bf P}'$ and setting $\bf P={\bf P}'\bf A$ we finish the proof. 
\end{proof}

\subsection{Free energy} We give the variational formula for the pressure of the model:
\begin{proposition}\label{Prop:pressure}
Let $c\geq \nu -1$. It holds
\be\label{eq:A-m'}
A(\b,\a,h)=\min_{m'\in\R^\nu}\left(\frac{\b(m')^2}{2}+\sum_a \a_aA_{CW}\left(\b c\a_a,\b h_a-\beta(\boldsymbol{A}^{-1}\boldsymbol{P^T}m')_a\right)\right)\,.
\ee
\end{proposition}

\begin{remark}
As $c=\nu-1$ the matrix ${\bf P}$ has rank $\nu-1$ and the functional in (\ref{eq:A-m'}) is independent on the first coordinate. This corresponds to a 1-parameter family of equivalent minima, and the minimisers are determined modulo a 1-parameter transformation. We can select a particular minimiser by taking any $c>\nu-1$ and then sending $c\to\nu-1$ (\ie requiring continuity in $c$). 
\end{remark}

\begin{proof}

Let $c\geq \nu-1$. We put $ m':=\boldsymbol{P}m$ (that is $m'=(\boldsymbol{v_1\cdot m},\dots, \boldsymbol{v_{\nu}\cdot m})$) and introduce the interpolating function 
\be
\Phi(s,t)=\frac 1 N \log \sum_{\boldsymbol{\sigma}}\exp \left\{N\left( t \sum_{a<b}\a_a\a_bm_am_b +\frac{(\beta-t)c}{2}\sum_a \a_a^2m_a^2 + \sum_a s_a m'_a, \right)\right\}\,.
\ee
$\Phi(s,t=0)$ reduces to the convex sum of the pressures of Curie Weiss models, each at inverse temperature $\b c\alpha_a$ and external field $s'_a:=(\boldsymbol{A}^{-1}\boldsymbol{P}^Ts)_a$.
For $t=\b$ and $s(h): \boldsymbol{P}^Ts(h)=\b {\bf A}h$, we recover our original model (${\bf A}$ was defined in the proof of the previous lemma). Because of Lemma \ref{lemma:radice}, differentiating in $t$ the interpolating function we get
\be\label{idea}
\partial_t\Phi(s,t)= \frac 1 2 \meanv{ 2\sum_{a<b}\a_a\a_bm_am_b -c\sum_a \a_a^2m_a^2}=-\frac{\meanv{(m,\boldsymbol{T}m)}}{2}=-\frac{1}2\sum_a\meanv{{m'_a}^2}\,.
\ee

In addition we have
\bea
\partial_{s_a}\Phi(s,t)&=&\meanv{m'_a},\nn\\
\Delta\Phi(s,t)&=&N\sum_a\left(\meanv{{m'_a}^2}-\meanv{m'_a}^2\right)\nn.
\eea
 
Since
$$
\Phi(s,0)=\sum_a \a_aA_{CW}\left(\beta c\a_a, s'_a\right),
$$
$\Phi$ satisfies a viscous Hamilton-Jacobi equation in $\R^\nu$
\be\label{eq:HJ}
\left\{
\begin{array}{rcl}
\partial_t\Phi(s,t)+\frac12|\nabla \Phi(s,t)|^2+\frac{1}{2N}\Delta \Phi(s,t)&=&0\\
\Phi(s,0)&=&\sum_a \a_aA_{CW}\left(\beta c\a_a,s'_a\right).
\end{array}
\right.
\ee
As $N\to\infty$, the solution of this PDE can be shown to approach the unique viscosity solution of the free Hamilton-Jacobi equation, given by the Hopf-Lax formula (see for instance \cite{can})
\be\label{eq:hopf-lax}
\Phi(s,t)=\min_{z\in\R^\nu}\left(\frac{(s-z)^2}{2t}+\Phi(z,0)\right)\,.
\ee
Let us introduce $z=s-m't$, so that we can rewrite the variational principle as
\be\label{eq:hopf-lax}
\Phi(s,t)=\min_{m'\in\R^\nu}\left(\frac{t(m')^2}{2}+\Phi(s-m't,0)\right)\,.
\ee
Thus
$$
A(\b,\a)=\Phi(s(h),\b)= \min_{m'\in\R^\nu}\left(\frac{\b(m')^2}{2}+\sum_a \a_aA_{CW}\left(\b c\a_a,\b h_a-\beta(\boldsymbol{A}^{-1}\boldsymbol{P^T}m')_a\right)\right)\,.
$$
\end{proof}
\begin{remark}
For $c>\nu-1$ we can express the pressure in terms of the physical magnetisation and of the energy tensor $\boldsymbol{J}^c=\boldsymbol{P^T P}$ as
\be\label{eq:A}
A(\b,\a)=\min_{m\in\R^\nu}\left(\frac{\b(m,\boldsymbol{J}^cm)}{2}+\sum_a \a_aA_{CW}\left(\b c\a_a,\b h_a-\beta (\boldsymbol{A}^{-1}\boldsymbol{J}^cm)_a\right)\right)\,.
\ee
\end{remark}

\begin{remark}
This variational principle is not separating the entropic and energetic contribution. This mirrors the fact that $(m,\boldsymbol{J}^cm)$ is not precisely the quadratic form associated to the energy (whence the $\min$ instead of a more familiar $\max$). There are two ways to pick out these two components, introducing the ancillary variable $M:=(M_1,\dots,M_\nu)\in\R^\nu$. One is to use the convexity principle of the CW pressure 
\be
A_{CW}(\b,h)=\max_{M}\left[-\b M^2+\phi(\b M+h)\right]
\ee
for each $a=1,\dots,\nu$ and get
\bea
A(\b,\a)&=&\min_{m\in\R^\nu}\max_{M\in\R^\nu}\left(\frac{\b c}{2}\sum_a\a_a^2(m^2_a-M^2_a)- \b\sum_{(a,b)} \a_a\a_bm_am_b\right.\nn\\
& +&\left. \sum_a \a_a \phi_a\left(\b c\a_a (M_a-m_a)+\b h_a+\beta\sum_{b\neq a} \a_b m_b\right)\right)\nn\,.
\eea

Optimisation leads to $m_a=M_a$ given by the self consistency relations (\ref{eq:selfcons-m}) below. Thus
\be\label{eq:A-energy-entropy}
A(\b,\a)= -\b\sum_{(a,b)} \a_a\a_bm_am_b + \sum_a \a_a \phi_a\left(\beta\sum_{b\neq a} \a_b m_b+\b h_a\right)\,.
\ee
This is not settled in a variational form, but for $\nu=2$ one can verify that (\ref{eq:A-energy-entropy}) corresponds to the $\min\max$ principle of \cite{bgg}\cite{gb2}.

%\begin{remark}
%For $\nu=2$ one can verify that (\ref{eq:A-energy-entropy}) corresponds to the $\min\max$ principle of \cite{bgg}\cite{gb2}.
%\end{remark}

Alternatively,  we can recover the entropic maximum principle from the one of the CW model (see the general framework presented in \cite{paris}). We introduce the entropy of each party as
\be
\SSS_a(M):=\lim_N\frac1N \log P(m_a\geq M)\,
\ee
and use the following formula \cite{paris}
\be\label{entropicoCW}
A_{CW}(\b,h)=\max_M\left(\frac{\b M^2}{2}+hM+\SSS(M)\right)\,. 
\ee
We plug (\ref{entropicoCW}) in (\ref{eq:A}) for each $a=1,\dots,\nu$ and obtain (after some manipulations)
\bea
A(\b,\a)&=&\min_{m\in\R^\nu}\max_{M\in\R^\nu}\left(\frac{\b c}{2}\sum_{a=1}^\nu \a_a^2(m_a^2-M_a^2)-\b\sum_{(a,b)}\a_a\a_bm_a(m_b-M_b)\right.\nn\\
&+&\left.\b\sum_{(a,b)}\a_a\a_bm_aM_b +\b\sum_a(\a_ah_aM_a+\a_a\SSS_a(M_a))\right)\,.\nn
\eea
As it can be directly checked, we can change the order of the $\max$ and the $\min$ in this expression. Therefore optimisation on $m$ yields again $m_a=M_a$. Thus we are left with
\be
A(\b,\a)=\max_{M\in\R^\nu}\left(\b\left(\sum_{(a,b)}\a_a\a_bM_aM_b +\sum_{a=1}^{\nu}\a_ah_aM_a\right)+\a_a\SSS_a(M_a) \right)\,.
\ee
This is the entropic principle for our model in the usual maximum formulation.
\end{remark}

\subsection{Magnetisations}
The minimisers of (\ref{eq:A-m'}) are $$m'_a(s,t)=-1/t\partial_{m'_a}(\Phi(s-m't,0))= \partial_{s_a}\Phi(s-m't,0)=\meanv{m'_a}_{s,t}\,,$$ \ie the thermal average of the rotated order parameter. Differentiating equation $(\ref{eq:A})$ w.r.t. $m$ we obtain
\be 
\boldsymbol{J}^c(m-\tilde{m}^{CW})=0,\quad \tilde{m}^{CW}_a= \left.\partial_x A^{CW}(\b c\a_a,x)\right|_{x=\b h_a-\beta (\boldsymbol{A}^{-1}\boldsymbol{J}^cm)_a}\,.
\ee
Since $\boldsymbol{J}^c$ is positive definite, we get $\nu$ self--consistent equations $m=\tilde{m}^{CW}$, where $\tilde{m}^{CW}_a$ satisfy
\begin{eqnarray} 
\tilde{m}^{CW}_a &=& \phi'_a\left(\b c\a_a\tilde{m}^{CW}_a+\b h_a- \beta (\boldsymbol{A}^{-1}\boldsymbol{J}^cm)_a\right) \nonumber \\
&=& \phi'_a\left(\b c\a_a\tilde{m}^{CW}_a+\b h_a- \beta c \alpha_a m_a-\beta \sum_{b\neq a} \alpha_b m_b\right)\,. 
\end{eqnarray}
Thus, the term depending on $c$ is canceled and we obtain
\be\label{eq:selfcons-m}
F_a(m,\b):=m_a -\phi'_a\left(\b h_a+\beta \sum_{b\neq a} \a_b m_b \right)=0\,.
\ee
Consider now the matrix
\be
M_{ab}:=\partial_{m_b} \phi'_a\left(\beta \sum_{c\neq a} \a_c m_c \right)\,.
\ee
Of course $\mathbb{I}-\b {\bf M}$ is the Jacobian matrix associated to the system (\ref{eq:selfcons-m}). It is convenient to put
$$
M^0_{ab}(\b):=\left.M_{ab}\right|_{m_1=\dots=m_\nu=0}=\begin{cases}
\b\a_b &a\neq b\,,\\
0&\mbox{otherwise}\,.
\end{cases}
$$
We are now ready to state the following

\begin{proposition}\label{prop:b<bc}
There is a $\b_c=\b_c(\a_1,\dots, \a_\nu)>0$ such that the equation
\be
\det(\mathbb I-{\bf M^0}(\b))=0
\ee
has a unique positive solution, and for $\b<\b_c$, $h_1=\dots=h_\nu=0$ we have $\meanv{m}=(0,\dots,0)$.
\end{proposition}

Moreover we can show that the critical temperature introduced in Proposition \ref{prop:b<bc} represents not only a bound for the paramagnetic phase, but it is actually a bifurcation point: a second order phase transition with the usual ferromagnetic critical exponent $\frac12$ occurs and the system magnetises. This is precisely stated in the subsequent

\begin{proposition}\label{prop:b>bc}
Let $\b>\b_c$ and $w\in \Ker(\mathbb I-{\bf M^0}(\b_c))$, $\|w\|=1$. Then there is a $m^*(\b, \a_1,\dots,\a_\nu)\in\R^\nu$ and a number $\kappa$, depending only on $\a_1,\dots,\a_\nu, P^1_4,\dots,P^\nu_4$ and $w$, such that
\be
m^*=w\sqrt{\frac{\b-\b_c}{\b_c^3\kappa}}+\OOO{\b-\b_c}
\ee
and
\be\label{eq:spontSymmBreak}
\lim_{\stackrel{(h_1,\dots,h_1)\to(0,\dots,0)}{(h,w)\lessgtr0}}\meanv{m}=\mp m^*\,,\quad \lim_{\stackrel{(h_1,\dots,h_1)\to(0,\dots,0)}{(h,w)=0}}\meanv{m}=0\,.
\ee
\end{proposition}
\begin{proof}[Proof of Proposition \ref{prop:b<bc}]
Since $\phi'(0)=0$ (\ref{eq:selfcons-m}) has always the trivial solution. According to the Gale-Nikaido theorem \cite{GN}, the set of equations (\ref{eq:selfcons-m}) has a global unique solution in the origin if all the principal minors of the Jacobian matrix are positive. Since $0\leq M_{ab}\leq M^0_{ab}$ and both matrices have all zero diagonal entries, it suffices to show that $\mathbb{I}-{\bf M^0}(\b)$ is positive definite.  

Let $2\leq\nu'\leq\nu$ and $I^{\nu'}\subseteq\{1,\dots,\nu\}$ with $|I^{\nu'}|=\nu'$. We define  $\mathcal{A}^k(I^{\nu'})$ the set of ordered $k$-ples in $I^{\nu'}$, \ie $a\in\mathcal{A}^k(I^{\nu'})$ is a multi-index $a=(a_1,\dots, a_k)$ with $a_i \in I^{\nu'}\subseteq\{1,\dots,\nu\}$, and we denote by $\a_a:=\a_{a_1}\dots\a_{a_k}$. We will use the formula $\det (A)=\sum_{\pi}\operatorname{sgn}(\pi)  a_{1\pi(1)}\ldots a_{n,\pi(n)} $. 

Each principal minor of $\mathbb{I}-\b M^0$ with rank $\nu'$ identified by $I^{\nu'}$, reads
$$
1+\sum_{k=1}^{\nu'}\mathcal{D}(k)(-1)^k \b^{k}\sum_{a\in \mathcal{A}^k(I^{\nu'})} \a_a\,,
$$  
where we group terms in the determinant according to the power of $\beta$, \ie to the number of fixed points in the permutation. 
$\mathcal{D}(k)$ is the sum (weighted with the relative sign) of all  the permutation of the set $(1,\ldots , k)$ with no fixed points, that can be computed as
$$
(-1)^k \mathcal{D}(k) = \det 
\begin{bmatrix}
0 & -1 & \cdots & -1 \\
-1 & 0 & \cdots & -1 \\
\vdots & \vdots & \ddots & \vdots \\
-1 & -1 & \cdots & 0
\end{bmatrix}
=\det (\boldsymbol{T}^{c=0})= 1-k,
$$  
since $\boldsymbol{T}^{c=0}$, already defined in the proof of Lemma 1, has eigenvalues $1$ with multiplicity $k-1$ and $1-k$ with multiplicity one.  
Thus $\b$ must satisfy the set of inequalities
\be\label{eq:beta-critico-nu'.k}
1-\sum_{k=2}^{\nu'}\b^{k}(k-1)\sum_{a\in \mathcal{A}^k(I^{\nu'})} \a_a\,>0\,,\quad \forall I^{\nu'}\subseteq\{1,\ldots,\nu\} .
\ee
Since for each $I^{\nu'}$ the r.h.s. of (\ref{eq:beta-critico-nu'.k}) defines a continuous and decreasing function of $\b\geq0$, there is a unique $\b_c(I^{\nu'})>0$ such that it is positive for $\b<\b_c(I^{\nu'})$ and negative otherwise. Therefore for
\be
\b<\min_{I^{\nu'}\subseteq\{1,\dots,\nu\}} \b_c(I^{\nu'})\,
\ee
the system of equations (\ref{eq:selfcons-m}) has a unique solution in $(0,\dots,0)$.
Finally, since all the terms in ($\ref{eq:beta-critico-nu'.k}$) are negative, the condition on the determinant suffices to determine the critical value of $\b$. 
\end{proof}

\begin{remark}
The Gale-Nikaido theorem used in the last proof can be interpreted in our context as follows: the $\nu$-partite model has no phase transition if for any $2\leq\nu'\leq\nu$, none of the $\nu'$-partite subsystems undergoes a phase transition (at zero fields). 
\end{remark}

\begin{proof}[Proof of Proposition \ref{prop:b>bc}]
From the analysis in the proof of Proposition \ref{prop:b<bc} we can infer that there exists a $w\in\R^\nu$ such that $\Ker(\mathbb{I}-{\bf M^0}(\b_c))=\{tw\in\R^\nu, t\in\R\}$. Moreover, looking at the second Fr\'echet derivative of the map $F(\b,m)$ defined in (\ref{eq:selfcons-m}) we get for any $a,b,c \in\{1,\dots,\nu\}$ 
\be
\partial_{m_bm_c}F_a\big|_{(\b,m)=(\b_c,0)}=-\b_c^2\a_b\a_c(1-\d_{ab})(1-\d_{ac})\phi^{(3)}(0)=0\,.
\ee
Therefore, since $P_4^{a}<0$, $a=1,\dots,\nu$, $(\b_c,0)$ is a supercritical (or pitchfork) bifurcation point for $F(\b,m)$ and the statement follows from standard results in bifurcation theory (for which we refer to \cite{AP}). 
\end{proof}

\subsection{The bipartite and tripartite models}
We give concrete examples for $\nu=2,3$. In the bipartite model condition (\ref{eq:beta-critico-nu'.k}) reduces to
\be\label{eq:bid}
1-\b^2\a_1\a_2>0,
\ee
and so $\b_c=1/\sqrt{\a_1\a_2}$. We can characterise quite precisely the spontaneous symmetry breaking for zero fields. It follows from a direct computation that the kernel of $\mathbb{I}-{\bf M^0}(\b_c)$ is spanned by the unitary vector
\be\label{eq:w-bip}
w:=\left(\begin{array}{c}\sqrt{\a_2}\\ \sqrt{\a_1}\end{array}\right)\,,
\ee
thus we have spontaneous symmetry breaking along the direction $w$ as
\be\label{eq:mbip}
m^*\sim\left(\begin{array}{c}\sqrt{\a_2}\sqrt{\b\sqrt{\a_1\a_2}-1}\\ \sqrt{\a_1}\sqrt{\b\sqrt{\a_1\a_2}-1}\end{array}\right)+\OOO{\b-\b_c}\,.
\ee
Heuristically, the proofs of Propositions 2,3 in this case are achieved by expanding ($\ref{eq:selfcons-m}$) around $m=0$:
\begin{eqnarray}\label{eq:biplin}
m_1-\phi'_1(\b \a_2 m_2)=m_1-\b \a_2 m_2 - \frac{P_1^4(0)}{3!}\b^3\a_2^3m_2^3 =\OOO{m_2^5}\,,\nn\\
m_2-\phi'_2(\b \a_1 m_1)=m_2-\b \a_1 m_1 - \frac{P_2^4(0)}{3!}\b^3\a_1^3m_1^3 =\OOO{m_1^5}\,.\nn
\end{eqnarray}
Solving w.r.t. $m_1$ the first equation and plugging the result in the second one, we get for $\b\sim\b_c$
$$
(1-\b^2\a_1\a_2)m_2= \frac{\b_c^3\a_2}{3!}(\a_1P^4_1(0)+\a_2P^4_2(0))m_2^3+\OOO{m^5}\,.
$$
This equation, solved for $m_2\neq 0$, gives (repeating the same argument for $m_1$) the scaling $(\ref{eq:mbip})$, with the proportionality constant in the form of Proposition \ref{prop:b>bc}. The role of $w$ is even clearer if we diagonalise ($\ref{eq:biplin}$), $\mathbb{I}-{\bf M^0}(\b_c)m = \OOO{m^3,(\b-\b_c)}$, with the transformation 
$m=\boldsymbol{O}m'$, with $\boldsymbol{O}=(w,w^\bot)$ such that
\be 
\boldsymbol{O}=
\boldsymbol{O^T}(\mathbb{I}-M^0(\b_c))\boldsymbol{O}=
\left(\begin{matrix}0 & 0 \\ 0 & \lambda   \end{matrix} \right).
\ee
In the diagonalised system $m_2'=\OOO{m_1^2}\sim 0$, that is $m\sim w$ and along the direction $w^\bot$ the magnetisation is weaker, according to the statement of Proposition \ref{prop:b>bc}.  

As for the tripartite model, the (three) principal minors of order two are
$
1-\b^2\a_1\a_2$,
$1-\b^2\a_1\a_3$,
$1-\b^2\a_2\a_3$,
while the determinant 
$
1-\b^2(\a_1\a_2+\a_1\a_3+\a_2\a_3)-2\b^3\a_1\a_2\a_3=0.
$
For $\beta$ approaching zero, they are all positive. If one of the minors  becomes negative, the same happens for the determinant. Thus ($\ref{eq:det}$) alone is a sufficient (and necessary) condition for the definition of the critical temperature (see Fig. \ref{fig:tctri}) as
\be\label{eq:det}
1-\b_c^2(\a_1\a_2+\a_1\a_3+\a_2\a_3)-2\b_c^3\a_1\a_2\a_3=0.
\ee 
Note that for $\a_{i}\to 0$, $i\in\{1,2,3\}$, condition ($\ref{eq:det}$) reduces to ($\ref{eq:bid}$). 
\begin{figure}
\includegraphics[scale=0.55]{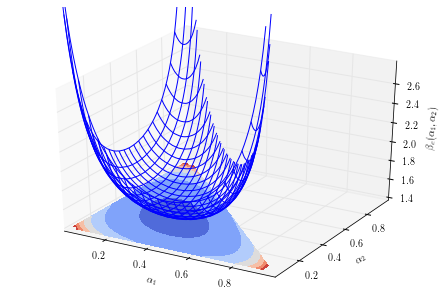}
\caption{Critical surface of the tripartite model as a function of $\a_1$ and $\a_2$.}\label{fig:tctri}
\end{figure}
Again we con compute directly the scaling of the spontaneous magnetisation by linearising ($\ref{eq:selfcons-m}$) around the paramagnetic solution as
\begin{eqnarray}\label{eq:triplin}
m_1-\b \a_2 m_2-\b \a_3 m_3 = \OOO{m^3}\,,\nn\\
m_2-\b \a_1 m_1-\b \a_3 m_3 = \OOO{m^3}\,,\nn\\
m_3-\b \a_1 m_1-\b \a_2 m_2 = \OOO{m^3}\,.\nn
\end{eqnarray}
Plugging the first equation in the second one, as soon as the minor of order two is different from zero, we have
\begin{eqnarray}
m_1 &=& \frac{\b\a_3+\b^2\a_2\a_3}{1-\b^2\a_1\a_2}m_3+\OOO{m^3}=\frac{(\mathbb{I}-M^0)_{\backslash 31}}{(\mathbb{I}-M^0)_{\backslash 33}}m_3+\OOO{m^3}\nn\\
m_2 &=& \frac{\b\a_3+\b^2\a_1\a_3}{1-\b^2\a_1\a_2}m_3+\OOO{m^3}=-\frac{(\mathbb{I}-M^0)_{\backslash 32}}{(\mathbb{I}-M^0)_{\backslash 33}}m_3+\OOO{m^3}\nn,
\end{eqnarray}
where $(\mathbb{I}-M^0)_{\backslash ij}$ denotes the minor obtained removing the $i$-th row and the $j$-th column. Hence, putting the first two equations into the third one, we obtain
$$
-\b\a_1(\mathbb{I}-M^0)_{\backslash 31}m_3 +\b\a_2(\mathbb{I}-M^0)_{\backslash 32}m_3-(\mathbb{I}-M^0)_{\backslash 33}m_3=\det(\mathbb{I}-{\bf M^0})m_3=\OOO{m^3}\,.
$$
Since $\det(\mathbb{I}-{\bf M^0}(\b))=(\beta-\b_c)P(\beta)$ we find the expected scaling $m_3\sim\sqrt{\b-\b_c}$ (and so for $m_1$ and $m_2$). Furthermore, from the Cramer decomposition, it is easy to see that the vector $m=((\mathbb{I}-M^0)_{\backslash 31},-(\mathbb{I}-M^0)_{\backslash 32},(\mathbb{I}-M^0)_{\backslash 32})$ spans the kernel of $\mathbb{I}-{\bf M^0}(\b)$.   

\section{Phase Diagram of the Bipartite Ferromagnet ($\nu=2$)}\label{SezioneBip}
 
\ni There is a precise equivalence between the bipartite models and a generalised ferromagnets studied in Appendix $\ref{sec:genferr}$, which we will often refer to. To see this duality, let us first introduce
\bea
u_1(x)&:=& \phi_2(x)=\log\mathbb{E}_\t e^{x\t}\label{eq:u1}\,,\\
u_2(x)&:=& \phi_1(x)=\log\mathbb{E}_\s e^{x\s}\label{eq:u2}\,.
\eea
\ni It can be verified at once that $u_1(\cdot)$ and $u_2(\cdot)$ are symmetric, non negative $C^2$ functions, uniformly convex, vanishing in the origin. Furthermore it is $u''_1(0)=\mathbb{E}[\t^2]=1$ and $u''_2(0)=\mathbb{E}[\s^2]=1$. Moreover, in virtue of our assumptions on the spin distribution ($\ref{cond:CW-subGauss}$), $u(x)- \frac{u''(0)x^2}{2}$ is negative and convex.

Now we write the partition function of the bipartite model in terms of $u_1$ or $u_2$ marginalising over one party. Since $\a_2=1-\a_1$, here we rename $\a_1=:\a$; moreover, with a little abuse of notations, we indicate as $\E_\s$ both expectations w.r.t. the spin configurations and the single spin variable. We have
\bea\label{eq:Z1}
Z_N(\beta, h_1, h_2)&=&\mathbb{E}_{\sigma^{1}}e^{\b N_1 h_1m_1}\left[\mathbb{E}_{\s^2}e^{(\beta \a m_1+\b h_2)\s^2}\right]^{N_2}\nn\\
&=&\mathbb{E}_{\sigma^{1}}e^{N(1-\a)\log\mathbb{E}_{\s^2} e^{(\a\beta m_1+\b h_2)\s^2}+\b N_1h_1m_1}\nn\\
&=&\mathbb{E}_{\sigma^{1}}e^{N_1\frac{1-\a}{\a} u_1(\a\beta m_1+\b h_2)+\b N_1 h_1m_1}\,.
\eea
Analogously one can write
\be\label{eq:Z2}
Z_N(\beta, h_1, h_2)=\mathbb{E}_{\s^{2}}e^{N_2\frac{\a}{1-\a} u_2((1-\a)\b m_2+\b h_1)+\b N_2h_2m_2}\,.
\ee

\ni Therefore we can regard the bipartite model as a generalised mono-partite ferromagnet with energy function given by $u_1$ (respectively $u_2$) and spin $\eta^1=\a\b\s^1+h_2$ (respectively $\eta^2=\b(1-\a)\s^2+h_1$) at inverse temperature $(1-\a)/\a$ (respectively $\a/(1-\a)$). Let us consider the representation (\ref{eq:Z1}).  Using Proposition $\ref{Prop:CW-FreeEnergy}$ the pressure of the model reads
\begin{eqnarray}\label{eq:MargFree}
A(\b,h_1,h_2)&=&A_1(M_1;\b,h_1,h_2) \\
&=&\frac{1-\a}{\a}[u_1(\a\b M_1+\b h_2)-\a\b M_1 u_1'(\a\b M_1 +\b h_2)]+u_2(\b(1-\a) u_1'(\a\b M_1 +\b h_2)+\b h_1)\nn
\end{eqnarray}
with  $M_1$ satisfying
\be\label{sceq}
M_1=u_2'(\beta(1-\alpha)u_1'(\alpha\beta M_1+\beta h_2)+\beta h_1).
\ee
A useful check consist in defining $M_2=u_1'(\alpha\beta M_1+\beta h_2)$, in order to rewrite $(\ref{sceq})$ as in $(\ref{eq:selfcons-m})$:
\begin{eqnarray}\label{bipsyst}
M_1&=&u_2'(\beta(1-\alpha)M_2+\beta h_1)\nonumber \\
M_2&=&u_1'(\alpha\beta M_1+\beta h_2)
\end{eqnarray}
Clearly the auxiliary variable $M_2$ is just the second order parameter, as we can see by marginalising on the first party at first. When $h_1=h_2=0$, checking the derivative of both sides of equation $(\ref{sceq})$ in $M_1=0$ we get the condition for the critical point
\be 
1=\beta^2\alpha (1-\alpha) u''_2(\beta(1-\alpha)u_1'(0))u_1''(0)=\beta^2\alpha (1-\alpha),
\ee 
since $u_{1,2}'(0)=0$. Putting  $\b_{c}:=(\a(1-\a) )^{-1/2}$, we find $M_1=0$ for $\b< \b_{c}$, while two other solutions $\pm M_1^*\neq0$ appear for $\b>\b_{c}$. In the same way, marginalising on the first party, we have $M_2=0$ for $\b<\b_{c}$  and spontaneous magnetisation  $\pm M_2^*\neq0$ for $\b>\b_c$, \ie the two partial magnetisations are synchronised.
The mapping within the generalised ferromagnet  and the use  of Proposition $\ref{pr:flu}$ allows also to characterise  the fluctuations of the magnetisations around the mean values. Below the critical temperature we have the following
\begin{proposition}\label{prop:Flutt-Bip}
Let $\beta\in [0,\b_c)$, $h=0$ and a random vector $\mathcal{X}_\b\sim \mathcal{N}(0,\boldsymbol{\chi}(\beta))$ with
\be 
\boldsymbol{\chi}(\beta)=
\left (
\begin{matrix}
\frac{\b^2_c}{\b^2_c-\beta^2} & \frac{\b\b_c}{\b_c^2-\beta^2}  \\
\frac{\b_c\b}{\b^2_c-\beta^2} & \frac{\b^2_c}{\b_c^2-\beta^2}  \\
\end{matrix}
\right)\,.
\ee
Then
\begin{equation}\label{eq:bivariato}
(\sqrt{N_1}m_1,\sqrt{N_2}m_2) \overset{d}{\longrightarrow} \mathcal{X}_\b\,.
\end{equation}
\end{proposition}
\begin{proof}
From $(\ref{eq:Z1})$ (with $h=0$) we can compute directly the moment generating function 
\begin{eqnarray}
\psi(M_1,M_2)&=&Z^{-1}\mathbb{E}_{\s^1,\s^2}\left[e^{M_1\frac{\sum_{i=1}^{N_1}\s^1_i}{\sqrt{N_1}}+M_2\frac{\sum_{i=1}^{N_2}\s^2_i}{\sqrt{N_2}}}e^{-\b H_N(\s^1,\s^2)}\right]\nonumber\\
&=&Z^{-1}\mathbb{E}_{\s^1}\left[e^{M_1\frac{\sum_{i=1}^{N_1}\s^1_i}{\sqrt{N_1}}}e^{N_1\frac{1-\a}{\a}u_1\left(\frac{\a\b}{N_1}\left(\sum_{i=1}^{N_1}\s^1_i +\frac{\sqrt{N_1}M_2}{(\b/\b_c)^2}\right)\right)}\right].
\end{eqnarray}
Following the proof of Proposition $\ref{pr:flu}$ on the fluctuations of a generalised ferromagnets at inverse temperature $\frac{1-\a}{\a}$ and energy $u(x)=u_1(\a\b x)$, we get $(\ref{eq:fg})$, \ie% with $\l=(\b/\b_c)^2$, \ie
$$
\psi(M_1,M_2)\to Z^{-1}e^{\frac{1}{2}\left(\frac{1}{1-(\b/\b_c)^4}\right)(M_1^2+M_2^2)+\left(\frac{(\b/\b_c)^2}{1-(\b/\b_c)^4}\right)M_1M_2},
$$ 
which gives (\ref{eq:bivariato}).
\end{proof}
As in the CW model, we see a different behaviour at the critical point. It is convenient to recall the definition of the magnetisation direction $w$ in (\ref{eq:w-bip}) and to denote the orthonormal vector as $w^\perp$.
\begin{proposition}\label{Prop:bip-fluct-crit}
Set $\b=\b_c$ and $h=0$. Moreover let us define the random vector $\mathcal{X}_c$ with 
\be\label{eq:density-Xc}
P(\mathcal{X}_c\in A):=Z^{-1}\int_{A} dv_1dv_2 e^{\frac{(\a P_4^1v_1^4+ (1-\a)P_4^2v_2^4)}{4!}}\d(v\cdot w^\perp)\,,
\ee
for any $A$, Borel set in $\R^2$, and $Z$ is a normalisation constant. It holds
\be\label{eq:bip-fluct-crit} 
N^{\frac14}(m_1,m_2)\overset{d}{\longrightarrow} \mathcal{X}_c\,.
\ee
\end{proposition}

\begin{proof}
First we notice
$$
\int\nu^{\frac{1-\a}{\a}}_{u_1}(dy)e^{xy}=\left(\E_{\s^2}\left[e^{x\s^2}\right]\right)^{\frac{1-\a}{\a}}\,,
$$
thereby $P_4^{u_1}=\frac{1-\a}{\a}P_4^2$. Then we follow the main steps of the proofs of Proposition $\ref{prop:Flutt-Bip}$ and \ref{pr:fluatc}. As $\b=\b_c$ we compute the joint moment generating function ($Z$ gives the proper normalisation in each step):
\bea
\psi(M_1,M_2)&=&Z^{-1}\mathbb{E}_{\s^1,\s^2}\left[e^{N^{\frac14}M_1m_1+N^{\frac14}M_2m_2+N\sqrt{\a(1-\a)}m_1m_2}\right]\nn\\
&=&Z^{-1}\mathbb{E}_{\s^1}\left[ e^{N^{\frac14}M_1m_1+N_1\frac{1-\a}{\a}u_1\left(\sqrt{\frac{\a}{1-\a}}m_1+\frac{\a^{3/4} }{1-\a}M_2 N_1^{-\frac34} \right)   }\right]\nn\\
&=&Z^{-1}\int \mu^{\frac{1-\a}{\a}}_{N_1}(dz)\E_{\s^1}\left[e^{\sum_i\s^1_i\left(\a^{-1/4}M_1N_1^{-\frac34}+\frac{z}{\sqrt{N_1}}\sqrt{\frac{\a}{1-\a}}\right)}\right]e^{M_2\frac{\a^{3/4} }{1-\a} zN_1^{-\frac14}}\nn\\
&=&Z^{-1}\int \mu^{\frac{1-\a}{\a}}_{N_1}\left(dz \frac{\sqrt{1-\a}}{\a^{1/4}}N_1^{\frac14}\right) e^{N_1\phi_1\left(\a^{-1/4}M_1N_1^{-\frac34}+\a^{1/4}zN_1^{-\frac14} \right)+zM_2\sqrt{\frac{\a}{1-\a}}}\nn\\
&=&Z^{-1}\int \mu^{1}_{N_1}\left(dz (\a N_1)^{\frac14}\right) e^{N_1\phi_1\left(\a^{-1/4}M_1N_1^{-\frac34}+\a^{1/4}zN_1^{-\frac14} \right)+zM_2\sqrt{\frac{\a}{1-\a}}}\nn\\
&\overset{N\to\infty}{\longrightarrow}&Z^{-1}\int dz \ e^{(\a P_4^1+(1-\a)P_4^2)\frac{z^4}{4!}+z(M_1+\sqrt{\frac{\a}{1-\a}}M_2)}\,,\nonumber 
\eea
which is the moment generating function of the density in (\ref{eq:density-Xc}).
\end{proof}

\vspace{1cm}

Switching on the external fields, we find different types of solutions to the equation (\ref{sceq}), arising from a more varied landscape for $A_{1}(m_1;\b,h_1,h_2)$ and $A_{2}(m_2;\b,h_1,h_2)$ in (\ref{eq:MargFree}). As usual the global maxima represent thermodynamically stable states, while the local ones are related to meta-stable states.

First we notice that $M_1=0$ is a solution of (\ref{sceq}) only if 
\be\label{eq:crit1}
h_1+ (1-\a)u'_1(\beta h_2)=0\,.
\ee  
This line separates the regions with positive and negative $M_1$ magnetisation in the plane $(h_1,h_2)$.  Analogously, we get for the other party the line 
\be\label{eq:crit2}
h_2+ \a u'_2(\beta h_1)=0\,
\ee
to divide regions of positive and negative $M_2$, see Fig. (\ref{fig1}). Note that for $\b<\b_c$ the two lines intersect only in $(h_1,h_2)=(0,0)$, while for $\b>\b_c$ other two intersection points appear, such that $h_1h_2 <0$. 

\begin{figure}
\includegraphics[width=0.45\textwidth]{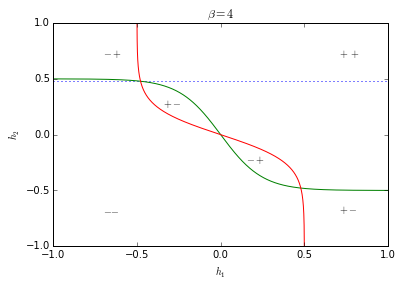}
\includegraphics[width=0.45\textwidth]{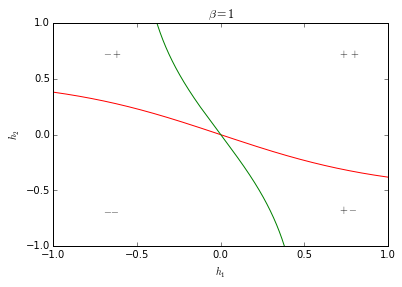}
\caption{Phase diagram in the plane $(h_1,h_2)$ for different values of $\b$: the two critical lines $h_1=- (1-\a)u'_1(\beta h_2)$ and $h_2=- \a u'_2(\beta h_1)$ separates  regions with positive and negative magnetisation $M_1$ (respectively $M_2$). Right panel: $\b=1<\b_c$ the only intersection is in the origin. Left panel: $\b=4>\b_c$, two non trivial intersection points appear at $h_2=\pm h_c^2$.  }\label{fig1}
\end{figure}

It is interesting to understand what happens along the critical lines.  At any fixed $\b$, we stay for example on the first critical line, parametrising it by $h_2$.  First we can look for the existence of non zero solutions by checking the derivative in zero of both sides of ($\ref{sceq}$), that yields 
\be\label{eq:cond1}
1<\b^2\a(1-\a)u''_1(\b h_2)\,.
\ee
Condition $(\ref{eq:cond1})$ defines a critical field $h_2=h_c^1(\b)$. This ensures the existence of two non-zero solutions of ($\ref{sceq}$) starting continuously from $M_1=0$ for $h_2<h_2^c$. It does not exclude the existence of such solutions for $h_2>h_2^c$, resulting from a first order phase transition, see Fig. \ref{fig2}. 

\begin{figure}
\includegraphics[width=0.5\textwidth]{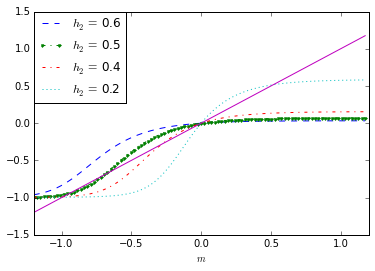}\\
\includegraphics[width=0.45\textwidth]{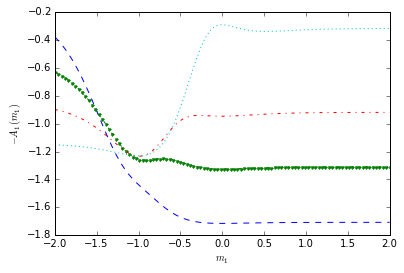}
\includegraphics[width=0.45\textwidth]{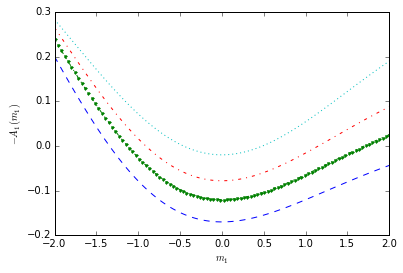}
\caption{On the first critical line $h_1=- (1-\a)u'_1(\beta h_2)$. Upper panel: representation of solutions of equation $(\ref{sceq})$ for $M_1$ at $\b=4>\b_c$ and different values of the field $h_2$. Bottom  panels: pressure $A_1(m_1;\b,h_1,h_2)$ for different values of the field $h_2$ and $\b=4$  (left), $\b=1$ (right).  For $\b>\b_c$ a first order phase transition at $m_1<0$ occurs for $h_2$ small enough, while for $\b<\b_c$ or high field $h_2$ the state $M_1=0$ is the only, thus stable, minimum. }\label{fig2}
\end{figure}

We have studied numerically the balanced ($\a=1/2$) bipartite model, with both parties made of Bernoulli $\pm1$ spins.  Fig. $\ref{fig2}$ shows the pressure  $A_1(m_1;\b,h_1,h_2)$, whose extremal points are solutions of  ($\ref{sceq}$), for two different values of $\b$. For $\b<\b_c$, $- A_1(m_1)$ has only one minimum in zero for each $(h_1,h_2)$ on the first critical line: note that, in this regime we only have the trivial intersection of the two critical lines, and equation $(\ref{eq:cond1})$ has no solution.  For  $\b>\b_c$ the landscape becomes richer by varying $(h_1,h_2)$ along the first transition line: high positive values of $h_2$ correspond to a single global minimum in $M_1=0$; lowering $h_2$ a local minimum appears at $M_1<0$, and lowering it further the minimum becomes global; finally a third local minimum $M_1\geq 0$ emerges continuously from zero. The situation is symmetric for negative field $h_2$.  In Fig. $\ref{fig4}$  all these minima are plotted as a function of $h_2$, distinguishing between local (metastable) and global (stable) states. It is possible to recognise three different critical values for the field $h_2\in\mathbb{R}^+$, $h_c^1<h_c^2<h_c^3$, such that
\bigskip
\begin{itemize}
\item if $h_2>h_c^3$ the only (global) minimum is at $M_1=0$;\\
\item if $h_2\in(h_c^2,h_c^3)$ $M_1=0$ is  the global minimum, but a local minimum $M_1=M_-< 0$ appears as a first order phase transition;\\
\item if $h_2\in(h_c^1,h_c^2)$ the minimum $M_1=0$ becomes  local while $M_-$  global; \\
\item if $h_2\in (0,h_c^1)$ a second order phase transition occurs with the emergence of a local minimum $M_1=M_+\geq 0$; $M_-$ is still the global minimum while $M_1=0$ becomes unstable (a maximum).
\end{itemize}
\bigskip
The second order phase transition at $h_2=h_c^1$ coincides with the bifurcation related to equation $(\ref{eq:cond1})$, although it corresponds always to a local minimum of the free energy, \ie a metastable state. 

Conversely, the first order phase transition related to $h_2=h_c^2$ can be better understood by checking the behaviour of the other party.  In fact it corresponds to the (non trivial) intersection point of the two critical lines of Fig. $\ref{fig2}$
\be 
\begin{cases}
h_1+ (1-\a)u'_1(\beta h_2)=0\\
h_2+ \a u'_2(\beta h_1)=0\,.
\end{cases}
\ee
Crossing $h_c^2$, the second party leaves the positive magnetisation region for the negative one, undergoing a first order phase transition in $M_2$ that reflects on the first party as well, see Fig $\ref{fig5}$.

\begin{figure}
\includegraphics[width=0.45\textwidth]{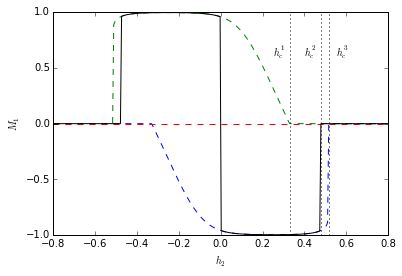}
\includegraphics[width=0.45\textwidth]{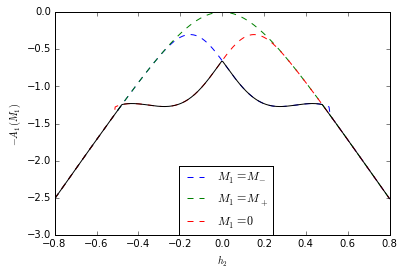}
\caption{Stable, metastable and unstable states as $\b=4>\b_c$. Left panel: solutions of equation $(\ref{sceq})$ for $M_1$ as a function of $h_2$. The black continuous line indicates the global, stable, solution as can be checked comparing the value of the corresponding pressure (right panel).  The vertical lines in the left panel mark the emergence of the phase transitions defined by the critical values $h_2=h_c^1,h_c^2,h_c^3$.}\label{fig4}
\end{figure}  
  
\begin{figure}
\includegraphics[width=0.5\textwidth]{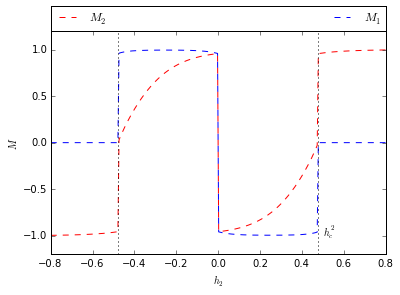}
\caption{magnetisations $M_1$ and $M_2$ on the first critical line $h_1=- (1-\a)u'_1(\beta h_2)$ and $\b=4>\b_c$ as a function of $h_2$. As soon as $M_2$ leaves the positive magnetisation region for the negative one, at $h_2=h_c^2$, a first order phase transition occurs, involving also $M_1$.  }\label{fig5}
\end{figure}
  
\section{Discussion}\label{Sect-Discussion}

Finally we point out some interesting problems not covered by our analysis. These mainly concern the phase transitions at $h_1,h_2\neq0$ that we have observed numerically for the bipartite model. We can ideally separate the problem of characterisation of these transitions from the study of the fluctuations of the magnetisations around them.

One possibility to study the phase transitions could be to look at the regularity of solutions of the Hamilton Jacobi equation (\ref{eq:HJ}). It has been shown in \cite{new} (see also \cite{russi}\cite{losanna}\cite{gb}) that in the CW model the phase transition corresponds to the formation of a shock wave in $(x=0, t>1)$ for the velocity field associated to (\ref{eq:HJ}). Note that we have skipped this approach since we could use the equivalence of bipartite and generalised ferromagnets, but it would be interesting to interpret the conditions (\ref{eq:crit1}) and (\ref{eq:crit2}) as shock lines. On the other hand the duality is at least not evident for multipartite ferromagnets, so definitely another approach would be helpful. 

Moreover, the same argument was used for finding the fluctuations of the bipartite model at zero external fields. Again some other ideas are needed to study the fluctuations in the case $\nu\geq3$, generalising the results of \cite{fedcon} to the non-convex models. 

In closing, here is a list of open questions we single out: first to prove mathematically the numerical results of Section \ref{SezioneBip} and possibly study the fluctuations of the magnetisations around these transitions; secondly, a further investigation of the multipartite models $\nu\geq3$ should include the fluctuations of the magnetisations at zero external fields; lastly it would be interesting to see the phase diagram for non zero external fields as $\nu\geq3$.

\bigskip

\subsection*{Acknowledgements} The authors thank F. Guerra for a fruitful correspondence on \cite{paris} and B. Schlein and M. Strani for some useful discussions. G. G. is supported by the Swiss National Science Foundation. D. T. is partially supported by "Avvio alla Ricerca 2014", Sapienza University of Rome, by the Centro Ennio De Giorgi, SNS (Pisa, Italia) and by GNFM, INdAM.

\newpage

\appendix

\section{Generalised Ferromagnets}\label{sec:genferr}

\ni We study here a generalisation of the CW model. Related results can be found in \cite{ee}, in more general settings. The achiements of this appendix are used in Section \ref{SezioneBip}. 

\ni The system in our interest is made by $N$ i.i.d. spin r.vs $\s_i$, $i=1,\dots,N$, with symmetric distribution satisfying (\ref{cond:CW-subGauss}). 
In addition let $u(x)$ be a symmetric, non negative, uniformly convex and smooth function with $u(0)=0$ and $\frac{u''(0)x^2}{2}\geq u(x)$. The Hamiltonian of the system in terms of the magnetisation reads
\be\label{eq:H}
H_N[u]:=-Nu(m)
\ee
and the associate partition function and pressure are respectively
\bea\label{eq:Z,A,f}
Z_N(\b, h)&:=&\mathbb{E}_{\s} e^{N\b u(m)+Nhm}\,,\nn\\
A_N(\b, h)&:=&\frac1N\log Z_N(\b, h)\nn\,.
\eea
As usual when we drop the subscript $N$ is to indicate the thermodynamic limit.
\subsection{Free Energy}
The limit form for the pressure is given by the following
\begin{proposition}\label{Prop:CW-FreeEnergy}
It holds
\be\label{eq:ACWTh}
A(\b,h)=\inf_N A_N(\b,h)=\max_{M\in\R}\left[ \b(u(M)-u'(M)M)+\phi(\b u'(M)+h) \right]\,.
\ee
\end{proposition}

\ni The proof is done by using the equivalence of ensembles (that is a standard strategy). We report it for completeness. 

\begin{proof}
Let $N_1+N_2=N$. The convexity of the interaction yields
$$
Z_N(\b,h)\leq Z_{N_1}(\b,h)Z_{N_2}(\b,h)\,,
$$
whence  the first equality in (\ref{eq:ACWTh}) follows. Now we introduce the trial partition function
\be
Z_N(\b,h;M):=\mathbb{E}_\s \exp\left[N\left(\b u(M)-\b u'(M)M+m\left(\b u'(M)+h\right)\right)\right].
\ee
We have $u(m)\geq u(M)+u'(M)(m-M)$, and so for any $M$
$$
Z_N(\b,h)\geq Z_N(\b,h;M)\,.
$$
Thus we put
\be
A(\b,h;M):= \lim_N\frac1N\log Z_N(\b,h; M)=\b(u(M)-u'(M)M)+\phi\left( \b u'(M)+h\right)
\ee
and since $A(\b,h;M)$ is concave in $M$, we find the bound
\be\label{eq:A>A(M)-MAX}
A(\b,h)\geq\max_M A(\b,h;M).
\ee

\ni Assume now for simplicity that $h\geq0$, so that the magnetisation must be non negative. Let us fix a $\bar M>0$. For $0=M_0<M_1<\dots<M_N=\bar M$ we use the partition $I_i:=[M_i,M_{i+1}]$ to write
\be\label{eq:Z-microcanonico}
Z_N(\b,h)=\E \left[\sum_{i=0}^{N-1} \chi(m\in I_i)e^{N\b u(m)+Nhm}\right]+\E \left[\chi(m>\bar M)e^{N\b u(m)+Nhm}\right]\,.
\ee
Using $u(x)-\frac{u''(0)x^2}{2}\leq0$ and $m^2+\frac{h^2}{4}\geq \frac{m^2}{2}+hm$, we can estimate the tail term as
\bea\label{eq:stimaL2L2Linf}
\E \left[\chi(m>\bar M)e^{N\b u(m)+Nhm}\right]&\leq& 
e^{N\frac{h^2}{4}} \E[\chi(m>\bar M)e^{\b u''(0)N m^2}]\nn\\
&\leq&\b u''(0)N e^{N\frac{h^2}{4}}\E[e^{(\b u''(0)+1) N m^2}] \b N\int_{\bar M^2}^\infty d \l e^{-N\l}\nn\\
&\leq&\b u''(0)N\left(r(\b u''(0)+1)e^{\frac{h^2}{4}-\bar M^2}\right)^N\,,
\eea
where we have also exploited Markov inequality, $N m^2\leq \sum_i\s_i^2$ and (\ref{cond:CW-subGauss}). 

\ni On the other hand in each $I_i$ we can pick a point $M^*_i$ so that up to a small error $e^{\bar M/N}$ we have
\bea
\E \left[\sum_{i=0}^{N-1} \chi(m\in I_i)e^{N\b u(m)+Nhm}\right]&\simeq&\mathbb{E}_{\s}\sum_{i=0}^{N-1} \d_{m,M^*_i} \exp\left[N\left(\b u(M^*_i)-\b u'(M^*_i)M^*_i+m\left(\b u'(M^*_i)+h\right)\right)\right]\nn\\
&\leq&\bar M\max_{M\in[0,\bar M]}Z_N(\b,h;M)\,.
\eea
These two bounds holds for any $\bar M$. Now let $\d>0$ and choose $\bar M=N^\d$. We get
\be
A_N(\b,h)\leq\d\frac{\log N}{N}+\frac1N\log\left(\max_{M\in[0,N^\d]} Z(\b,h;M)+\OOO{e^{-N^{1+2\d}}}\right)\,,
\ee
which, as $N\to\infty$, leads to
\be
A_N(\b,h)\leq\max_{M\in\R} A(\b,h;M)\,,
\ee
that completes the proof.
\end{proof}

\begin{remark}
The assumption $u(x)\leq u''(0)x^2/2$ is used only in the estimate of the tail term of (\ref{eq:Z-microcanonico}). Therefore it turns out to be immaterial if we consider only compactly supported spin distributions (as $\pm1$ spin), for which formula (\ref{eq:ACWTh}) remains valid also without this hypothesis (this corresponds somehow to allow $p$-spin interactions). For general spin  distributions with fast enough decay at infinity, one needs some further assumption on the energy: for instance, we know that quadratic energy balances well spin distribution as in (\ref{cond:CW-subGauss}) \cite{ellis}. We are not interested here in identifying precisely this relation for generic $u(x)$, since $u(x)\leq u''(0)x^2/2$ is enough to study bipartite models (see Section \ref{SezioneBip}).
\end{remark}

\subsection{Phase Transition and Fluctuations}
Optimising (\ref{eq:ACWTh}) we find
$$
\b u''(M)\left(M-\phi'(\b u'(M)+h)\right)=0\,,
$$
that immediately implies
\be\label{eq:M-mono-self}
M=\phi'(\b u'(M)+h)\,.
\ee
\ni This is the standard self consistent equation for ferromagnets. For $h\neq0$ we have always a sole solution; as $h=0$, $M=0$ still is a solution, but we find other two symmetric solutions if
$$
\partial_M\left[\frac{\mathbb{E}_\s \s e^{\s(\b u'(M)+h)} }{\mathbb{E}_\s e^{\s(\b u'(M)+h)} }\right]_{M=0}=\b u''(0)>1\,.
$$
Thus the critical point is
\be\label{eq:critical}
\b_c=\frac{1}{u''(0)}.
\ee

\ni As a last step of our analysis, we study the fluctuations around the transition point. We introduce the magnetic susceptivity
\be
\chi(\b,h):=\partial_hM(\b,h)= \lim_{N\to\infty}N \meanv{(m - M(\b,h))^2}\,.
\ee 
Setting
$$
V(\b,h):=\frac{\mathbb{E}_\s \s^2 e^{\s(\b u'(M)+h)} }{\mathbb{E}_\s e^{\s(\b u'(M)+h)} }-\left(\frac{\mathbb{E}_\s \s e^{\s(\b u'(M)+h)} }{\mathbb{E}_\s e^{\s(\b u'(M)+h)} }\right)^2
$$
one sees immediately that
$$
\chi(\b,h)=\frac{V(\b,h)}{1-\b u''(M)V(\b,h)}\,.
$$
We shorten $\chi(\b,0):=\chi(\b)$, and we notice $\chi(0)=\mathbb{E}[\s^2]$. Thus, as long as $M(\b,h)=0$ it must be
$$
\chi(\b)=\frac{1}{1-\b u''(0)},
$$
\ie  the susceptivity is divergent at $\b_c$ where spontaneous magnetisation occurs. 
We have the following
\begin{proposition}\label{pr:flu}
Let $\beta\in [0,\beta_c)$ and $h=0$. Then
\begin{equation}
\sqrt{N}m_N(\sigma)\overset{d}{\longrightarrow} \mathcal{N}(0,\chi(\beta)) 
\end{equation}
\end{proposition}
The proof of Proposition $\ref{pr:flu}$ and \ref{pr:fluatc} (below) is based on the following representation. For any such $u$ there is a probability measure $\nu^\b_u$ which is conjugate to $u$ by Laplace transform
\be\label{eq:dual}
e^{\b u(x)}=\int\nu^\b_u(dy)e^{yx}\,.
\ee
This correspondence is one to one. One sense is obvious. To see that for any convex function there is a probability measure such that (\ref{eq:dual}) holds is much less direct. It is a consequence of Bernstein's theory of completely monotonic functions, as $e^{\b u(x)}$ turns out to be (see \cite{fe}, XIII.4).
Now, consider $\xi_1,\dots\xi_N$ i.i.d. r.vs with $\xi_1\sim \nu^\b_u$ and $X_N:=N^{-\frac12}\sum_{i=1}^N\xi_i$. Then it is defined the measure $\mu^\b_{u,N}$ by $X_N\sim\mu^\b_{u,N}$. By the central limit theorem 
\be\label{lem:mu1}
X_N\overset{d}{\longrightarrow} \mathcal{N}(0,\beta u''(0))\,,\qquad\mbox{or}\qquad\mu^{\beta}_{u,N} \overset{w}{\longrightarrow} \frac{e^{-\frac{x^2}{2\b u''(0)}}}{\sqrt{2\pi\b u''(0)}}\,.
\ee
Note that $\mu_N^{\b_c}$ approaches the standard Gaussian as $N\to\infty$. Thus we have
\be\label{lem:mu}
e^{\b N u\left(\frac{x}{\sqrt N}\right)}=\left(\E_{\xi}e^{\frac{\xi x}{\sqrt{N}}}\right)^N=
\E_{\xi_1}\dots\E_{\xi_N} e^{x\frac{\sum_{i=1}^N\xi_i}{\sqrt{N}}}
=\int\mu^\b_N(dy)e^{yx}\,,
\ee
whence we can also see that $e^{\b N u\left(\frac{x}{\sqrt N}\right)}\longrightarrow e^{\b u''(0)x^2/2}$ as $N\to\infty$. 

\begin{proof}[Proof of Proposition $\ref{pr:flu}$]
Using  $(\ref{lem:mu})$  we can compute the moment generating function of $\sqrt{N}m(\s)$. For a more general result we consider
\begin{eqnarray}
\psi_{u,N}(M_1,M_2)&=& \E_\s\left[Z^{-1}e^{\b N u(\frac{\sum_{i=1}^N\s_i +\left(\sqrt{N}\frac{M_2}{\b u''(0)}\right)}{N})}e^{M_1\frac{\sum_{i=1}^N\s_i}{\sqrt{N}}}\right]\nn\\
&=&Z^{-1}\int d\mu^{\beta}_{u,N}(s) e^{\frac{sM_2}{\b u''(0)}}\E_\s\left[e^{\left(\frac{s+M_1}{\sqrt{N}}\right)\sum_{i=1}^N\s_i}\right]\nonumber\\
&=& Z^{-1}\int d\mu^{\beta}_{u,N}(s)e^{\frac{sM_2}{\b u''(0)}} e^{N\phi\left(\frac{s+M_1}{\sqrt{N}}\right)}
=Z^{-1}\int d\mu^{\beta}_{u,N}(s)e^{\frac{sM_2}{\b u''(0)}} e^{\frac{(s+M_1)^2}{2}}+O(1/N)\nn\,.
\end{eqnarray}
As soon as $\int d\mu^{\beta}_{u,N}(s) e^{\frac{s^2}{2}}<\infty$ holds, \ie, from $(\ref{lem:mu})$, for $\beta\in[0,\b_c)$, and using $(\ref{lem:mu1})$ 
\begin{eqnarray}\label{eq:fg}
\psi_{u,N}(M_1,M_2) &\stackrel{N\to\infty}{\longrightarrow}& 
e^{\frac{M_1^2}{2}}\ Z^{-1}\int ds\  e^{\frac{s^2}{2}\left(\frac{1-\b u''(0)^2}{\b u''(0)^2}\right)} e^{s\left(M_1+\frac{M_2}{\b u''(0)}\right)}\nonumber \\
&=&Z^{-1}e^{\frac{1}{2}\left(\frac{1}{1-\b u''(0)^2}\right)(M_1^2+M_2^2)+\left(\frac{\b u''(0)}{1-\b u''(0)^2}\right)M_1M_2}\,.
\end{eqnarray}
Thus $\psi_{u,N}(M_1,M_2)$ is the moment generating function of a bivariate Gaussian random variable. For $M_2=0$ we get the generating function of $\sqrt{N}m(\s)$, whence we conclude that its distribution must be a centred Gaussian with variance $\chi(\b)$.
\end{proof}

Lastly, at the critical point we have
\begin{proposition}\label{pr:fluatc}
Let $\beta=\beta_c$, $h=0$,  and 
$$
P^u_4:=\left[\frac{d^4}{dx^4}\log\left(\int\nu_u^\b(dy)e^{xy}\right)\right]_{x=0}\,.
$$
Moreover let $\mathcal{X}_c$ a r.v. with probability density given by $$\frac1Ze^{(P_4+P_4^u)\frac{x^4}{4!}}$$ ($Z$ is a normalisation coefficient). It holds
\be 
N^{1/4}m_N(\sigma)\overset{d}{\longrightarrow} \mathcal{X}_c\,.
\ee
\end{proposition}
\begin{remark}
The condition $u(x)\leq u''(0)x^2/2$ yields $P_4^u<0$, while $P_4<0$ by our initial assumptions.
\end{remark}
\begin{proof}
As in the proof of Proposition $\ref{pr:flu}$, we compute the moment generating function of $N^{1/4}m_N(\sigma)$
\begin{eqnarray}
\psi_u(M)&=& \E_\s[Z^{-1}e^{\b_c N u\left(\frac{\sum_{i=1}^N\s_i}{N}\right)}e^{M\frac{\sum_{i=1}^N\s_i}{N^{3/4}}}]=
Z^{-1}\int \mu^{\beta_c}_{u,N}(dz) \E_\s[e^{\left(\frac{s}{\sqrt{N}}+\frac{M}{N^{3/4}}\right)\sum_{i=1}^N\s_i}]\nonumber\\
&=& Z^{-1}\int \mu^{\beta_c}_{u,N}(dz) e^{N\phi\left(\frac{z}{\sqrt{N}}+\frac{M}{N^{3/4}}\right)}
= Z^{-1}\int \mu^{\beta_c}_{u,N}(dz N^{1/4}) e^{N\phi\left(\frac{z}{N^{1/4}}+\frac{M}{N^{3/4}}\right)}\nonumber \\
&=& Z^{-1}\int \mu^{\beta_c}_{u,N}(dzN^{1/4})e^{\sqrt{N}\frac{z^2}{2}} e^{zM} e^{\frac{P_4}{4!}z^4}+ o(1)\nonumber \\
&\stackrel{N\to\infty}{\longrightarrow}& Z^{-1}\int dz e^{zM} e^{\frac{P_4+P^u_4}{4!}z^4}.
\end{eqnarray}
In the last line the crucial observation is that $\mu^{\beta_c}_{u,N}(ds N^{\frac14})e^{\sqrt{N}\frac{z^2}{2}}$ tends weakly to the Lebesgue measure times $e^{\frac{P^{u}_4z^4}{4!}}$, as a consequence of the moderate deviations from central limit theorem (see \cite{pr}, VIII.2). 
\end{proof}

\end{document}